\documentclass[copyright,noncommercial]{eptcs}

\usepackage{amsmath,latexsym,amsfonts,amssymb,amsthm}
\usepackage{rotating,graphs,stmaryrd}

\providecommand{\event}{}

\providecommand{\firstpage}{1}

\newtheorem{lemma}{Lemma}
\newtheorem{proposition}{Proposition}

\theoremstyle{definition}

\theoremstyle{remark}

\newcommand{\tuple}[1]{\langle#1\rangle}

\newcommand{\ttt}{\texttt}

\newcommand{\mtt}{\mathtt}
\newcommand{\mrm}{\mathrm}

\newcommand{\N}{\mathbb{N}}

\newcommand{\Sem}[1]{\llbracket{#1}\rrbracket}

\newcommand{\x}{\mathtt{x}}
\newcommand{\y}{\mathtt{y}}

\newcommand{\rednodecolour}{\graphnodecolour(1,.03,.03)}

\newcommand{\grey}{\graphnodecolour{.75}}
\newcommand{\dashed}{\graphlinedash{4 2}}

\newcommand{\dder}{\Rightarrow}
\newcommand{\ldder}{\Leftarrow}
\newcommand{\der}{\Rightarrow^*}
\newcommand{\lder}{\Leftarrow^*}

\newcommand{\DSldder}{\mathop{\Leftarrow}\limits}

\title{Reasoning about Graph Programs}

\author{Detlef Plump 
\institute{The University of York, United Kingdom}}

\def\titlerunning{Reasoning about Graph Programs}
\def\authorrunning{D. Plump}

\begin{document}

\maketitle
\thispagestyle{empty}

\begin{abstract}
GP 2 is a non-deterministic programming language for computing by graph transformation. One of the design goals for GP 2 is syntactic and semantic simplicity, to facilitate formal reasoning about programs. In this paper, we demonstrate with four case studies how programmers can prove termination and partial correctness of their solutions. We argue that GP 2's graph transformation rules, together with induction over the length of program executions, provide a convenient framework for program verification.
\end{abstract}

\section{Introduction}
\label{sec:introduction}

The use of graphs to model dynamic structures is ubiquitous in computer science: application areas include compiler construction, pointer programming,  model-driven software development, and natural language processing. The behaviour of systems in such domains can be captured by graph transformation rules specifying small state changes. Current languages based on graph transformation rules include 
AGG \cite{Runge-Ermel-Taentzer11a},
GReAT \cite{Agrawal-Karsai-Neema-Shi-Vizhanyo06a},
GROOVE \cite{Ghamarian-deMol-Rensink-Zambon-Zimakova12a},
GrGen.Net \cite{Jakumeit-Buchwald-Kroll10a} and
PORGY \cite{Fernandez-Kirchner-Mackie-Pinaud14a}.
This paper focusses on the graph programming language GP 2 \cite{Plump12a} which aims to support formal reasoning about programs.\footnote{GP stands for \emph{graph programs}.}

A rigorous Hoare-logic approach to verifying programs in GP 1 (the predecessor of GP 2) is described in \cite{Poskitt-Plump12a}. However, this calculus is restricted to programs in which loop bodies and guards of branching statements are sets of rules rather than arbitrary subprograms. Moreover, the assertions of \cite{Poskitt-Plump12a} are first-order formulas and hence cannot express non-local properties such as connectedness or the absence of cycles. (Such properties can be expressed with the monadic second-order assertions of \cite{Poskitt-Plump14a}. That paper's framework is currently extended to GP 2.)

In this paper, we take a more relaxed view on program verification and express specifications and proofs in ordinary mathematical language. Besides lifting the mentioned restrictions, this approach allows programmers to formulate invariants and induction proofs succinctly, without getting stuck in formal details. The possible reduction in rigor need not be a drawback if the liberal approach precedes and complements rigorous verification in a formal calculus such as Hoare-logic.

In Section \ref{sec:transitive} to Section \ref{sec:series-parallel}, we verify four simple GP 2 programs, for generating the transitive closure of a graph, computing a vertex colouring, and checking graph-theoretic properties. In all case studies, we prove termination and partial correctness of the program in question. It turns out that graph transformation rules together with induction on derivations provide a convenient formalism for reasoning.

\section{The Language GP 2}
\label{sec:gp2}

We briefly introduce GP 2, by describing some selected features and showing an example of a graph transformation rule. The original definition of GP 2, including a formal operational semantics, is given in \cite{Plump12a}; an updated version can be found in \cite{Bak15a}. There are currently two implementations of GP 2, a compiler generating C code \cite{Bak-Plump16a} and an interpreter for exploring the language's non-determinism \cite{Bak-Faulkner-Plump-Runciman15a}. 

The principal programming constructs in GP~2 are conditional graph-transformation rules labelled with expressions. Rules operate on \emph{host graphs}\/ whose nodes and edges are labelled with lists of integers and character strings. Variables in rules are of type \texttt{int}, \texttt{char}, \texttt{string}, \texttt{atom} or \texttt{list}, where \texttt{atom} is the union of \texttt{int} and \texttt{string}. Atoms are considered as lists of length one, hence integers and strings are also lists. Given lists $\mtt{x}$ and $\mtt{y}$, their concatenation is written \texttt{x:y}. 

Besides a list, labels may contain a \emph{mark}\/ which is one of the values \ttt{red}, \ttt{green}, \ttt{blue}, \ttt{grey} and \ttt{dashed} (where \ttt{grey} and \ttt{dashed} are reserved for nodes and edges, respectively). Marks may be used to  highlight items in input or output graphs, or to record which items have already been visited during a graph traversal. In this paper, we assume that programs are applied to input graphs without any marks. This allows to formulate succinct correctness claims in the case studies below.

Figure \ref{fig:rule-example} shows an example of a rule which replaces the grey node and its incident edges with a dashed edge labelled with the integer 7. In addition, nodes 1 and 2 are relabelled with the values of $\mtt{x{:}y}$ and $\mtt{n \mathop\ast n}$, respectively, where the actual parameters for $\mtt{x}$, $\mtt{y}$ and $\mtt{n}$ are found by matching (injectively) the left-hand graph in a host graph. The rule is applicable only if the grey node is not incident with other edges (the \emph{dangling condition}) and if the \ttt{where}-clause is satisfied. The latter requires that $\mtt{n}$ is instantiated with a negative integer and that there is no edge from node 1 to node 2 in the host graph.

\begin{figure}[htb]
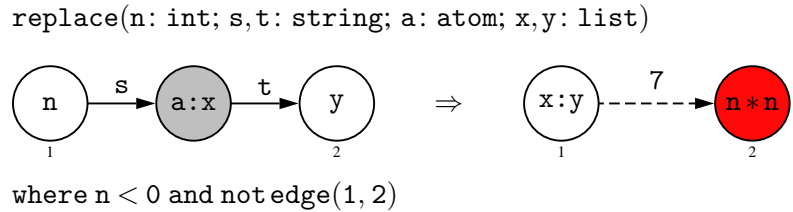

\centering
   \setlength{\unitlength}{0.5in}
   \graphnodesize{0.8}
   \grapharrowlength{0.2}
   \grapharrowwidth{0.6}
   \autodistance{1.0}
   \graphnodecolour{1}
\begin{tabular}{l}
$\mtt{replace (n\colon int;\; s,t\colon string;\; a\colon atom;\; x,y\colon list)}$
\\[1ex]
\begin{graph}(3.8,1.44)(.2,.83)
 \roundnode{n1}(.6,1.7)
  \autonodetext{n1}{\texttt{n}}\autonodetext{n1}[s]{\tiny 1}
 \roundnode{n2}(2.1,1.7)[\grey]
  \autonodetext{n2}{\texttt{a:x}}
 \roundnode{n3}(3.6,1.7)
  \autonodetext{n3}{\texttt{y}}\autonodetext{n3}[s]{\tiny 2}
 \diredge{n1}{n2}\bowtext{n1}{n2}{.12}{\texttt{s}}
 \diredge{n2}{n3}\bowtext{n2}{n3}{.12}{\texttt{t}}
\end{graph}
\begin{graph}(1.4,1.44)(0,-.1)
 \freetext(.7,.75){$\dder$}
\end{graph}
\begin{graph}(2.8,1.44)(.2,.83)
 \roundnode{n1}(.6,1.7)
  \autonodetext{n1}{\texttt{x:y}}\autonodetext{n1}[s]{\tiny 1}
 \roundnode{n3}(2.6,1.7)[\rednodecolour]\autonodetext{n3}[s]{\tiny 2}
  \autonodetext{n3}{$\mtt{n \ast n}$}
 \diredge{n1}{n3}[\dashed]\bowtext{n1}{n3}{.12}{$\mtt{7}$}
\end{graph}
\\[-1.5ex]
\ttt{where} $\mtt{n < 0}$ \texttt{and} $\mtt{not\, edge(1,\,2)}$ 
\end{tabular}

\caption{Declaration of a conditional rule} \label{fig:rule-example}
\end{figure}

The grammar in Figure \ref{fig:command-syntax} gives the abstract syntax of GP 2 programs (omitting rule declarations and graph labels). A program consists of declarations of conditional rules and procedures, and exactly one declaration of a main command sequence. The category RuleId refers to declarations of conditional rules in RuleDecl. Procedures must be non-recursive, they can be seen as macros with local declarations.

\begin{figure}[htb]
\begin{center}
\begin{tabular}{lcl}
Prog & ::= & Decl \{Decl\} \\
Decl & ::= & RuleDecl $\mid$ ProcDecl $\mid$ MainDecl \\
ProcDecl & ::= & ProcId `=' [ `[' LocalDecl `]' ] ComSeq \\
LocalDecl & ::= & (RuleDecl $\mid$ ProcDecl) \{LocalDecl\} \\
MainDecl & ::= & \ttt{Main} `=' ComSeq \\
ComSeq & ::= & Com \{`;' Com\} \\
Com & ::= & RuleSetCall $\mid$ ProcCall \\
&& $\mid$ \ttt{if} ComSeq \ttt{then} ComSeq [\ttt{else} ComSeq] \\
&& $\mid$ \ttt{try} ComSeq [\ttt{then} ComSeq] [\ttt{else} ComSeq] \\
&& $\mid$ ComSeq `{!}' $\mid$ ComSeq \ttt{or} ComSeq \\
&& $\mid$ `(' ComSeq `)' $\mid$ \ttt{break} $\mid$ \ttt{skip} $\mid$ \ttt{fail} \\
RuleSetCall & ::= & RuleId $\mid$ `\{' [RuleId \{`,' RuleId\}] `\}' \\
ProcCall & ::= & ProcId 
\end{tabular}
\end{center}
\caption{Abstract syntax of GP 2 programs \label{fig:command-syntax}}
\end{figure}

The call of a rule set $\{r_1,\dots,r_n\}$ non-deterministically applies one of the rules whose left-hand graph matches a subgraph of the host graph such that the dangling condition and the rule's application condition are satisfied. The call \emph{fails}\/ if none of the rules is applicable to the host graph. 

The command \ttt{if} $C$ \ttt{then} $P$ \ttt{else} $Q$ is executed on a host graph $G$ by first executing $C$ on a copy of $G$. If this results in a graph, $P$\/ is executed on the original graph $G$; otherwise, if $C$ fails, $Q$ is executed on $G$. The \ttt{try} command has a similar effect, except that $P$\/ is executed on the result of $C$'s execution. 

The loop command $P!$ executes the body $P$\/ repeatedly until it fails. When this is the case, $P!$ terminates with the graph on which the body was entered for the last time. The \ttt{break} command inside a loop terminates that loop and transfers control to the command following the loop.

A program $P$ \ttt{or} $Q$ non-deterministically chooses to execute either $P$ or $Q$, which can be simulated by a rule-set call and the other commands \cite{Plump12a}. The commands \ttt{skip} and \texttt{fail} can also be expressed by the other commands.

\section{Case Study: Transitive Closure}
\label{sec:transitive}

A graph is \emph{transitive} if for each directed path from a node $v$ to another node $v'$, there is an edge from $v$ to $v'$. The program \ttt{TransClosure} in Figure \ref{fig:transitive-closure} computes the transitive closure of a host graph $G$\/ by applying the single rule \ttt{link} as long as possible. Each application amounts to non-deterministically selecting a subgraph of $G$\/ that matches \ttt{link}'s left graph, and adding to it an edge from node 1 to node 3 provided there is no such edge (with any label). Figure \ref{fig:TransClosure_execution} shows an execution of \ttt{TransClosure} in which \ttt{link} is applied eight times (arcs with two arrowheads represent pairs of edges of opposite direction). The resulting graph is the complete graph of four nodes because the start graph is a cycle. 

\begin{figure}[htb]
\begin{center}
 \input{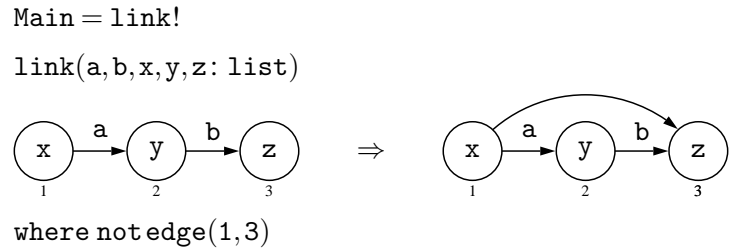}
\end{center}
\caption{The program \protect\ttt{TransClosure} \label{fig:transitive-closure}}
\end{figure}

The next two propositions show that \ttt{TransClosure} is correct: for every input graph $G$, the program produces the smallest transitive graph that results from adding unlabelled edges to $G$.\footnote{By a graphical convention of GP 2, ``unlabelled'' nodes and edges are actually labelled with the empty list.} 

\begin{figure}[htb]
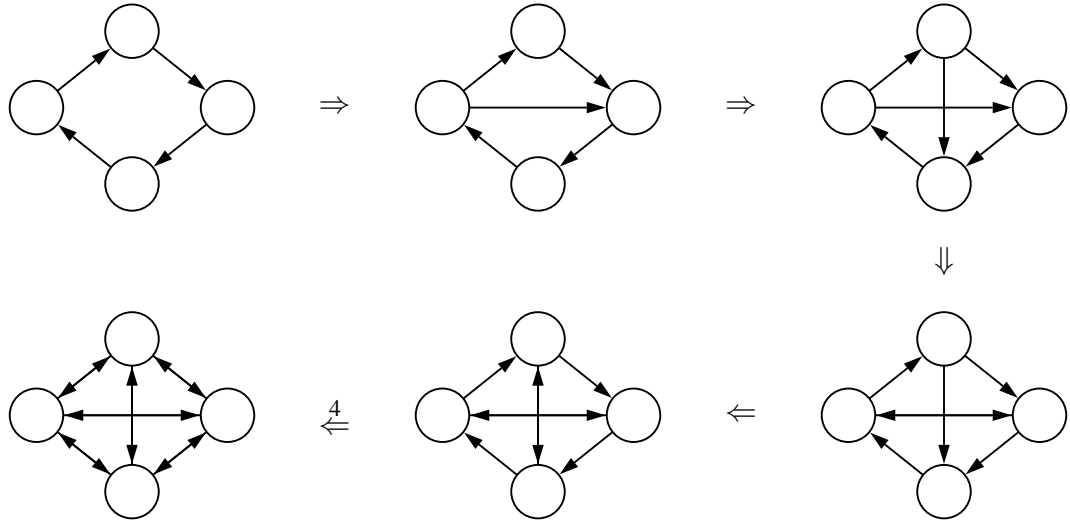

\begin{center}
    \setlength{\unitlength}{0.5in}
   \graphnodesize{0.58}
   \grapharrowlength{0.2}
   \grapharrowwidth{0.6}
   \autodistance{1.0}
   \graphnodecolour{1}

\begin{tabular}{c}
\begin{graph}(3.4,2.4)(.4,0)
 \roundnode{n1}(1.1,1.2)
 \roundnode{n2}(2.1,2)
 \roundnode{n3}(3.1,1.2)
 \roundnode{n4}(2.1,.4)
 \diredge{n1}{n2}
 \diredge{n2}{n3}
 \diredge{n3}{n4}
 \diredge{n4}{n1}
\end{graph}
\begin{graph}(.7,2.4)
 \freetext(.35,1.2){\large $\dder$}
\end{graph}
\begin{graph}(3.4,2.4)(.4,0)
 \roundnode{n1}(1.1,1.2)
 \roundnode{n2}(2.1,2)
 \roundnode{n3}(3.1,1.2)
 \roundnode{n4}(2.1,.4)
 \diredge{n1}{n2}
 \diredge{n2}{n3}
 \diredge{n3}{n4}
 \diredge{n4}{n1}
 \diredge{n1}{n3}
\end{graph}
\begin{graph}(.7,2.4)
 \freetext(.35,1.2){\large $\dder$}
\end{graph}
\begin{graph}(3.4,2.4)(.4,0)
 \roundnode{n1}(1.1,1.2)
 \roundnode{n2}(2.1,2)
 \roundnode{n3}(3.1,1.2)
 \roundnode{n4}(2.1,.4)
 \diredge{n1}{n2}
 \diredge{n2}{n3}
 \diredge{n3}{n4}
 \diredge{n4}{n1}
 \diredge{n1}{n3}
 \diredge{n2}{n4}
\end{graph}
\\
\begin{graph}(3.4,.6)
\end{graph}
\begin{graph}(.7,.6)
\end{graph}
\begin{graph}(3.4,.6)
\end{graph}
\begin{graph}(.7,.6)
\end{graph}
\begin{graph}(3.4,.6)
 \freetext(1.7,.3){\large $\Downarrow$}
\end{graph}
\\
\begin{graph}(3.4,2.4)(.4,0)
 \roundnode{n1}(1.1,1.2)
 \roundnode{n2}(2.1,2)
 \roundnode{n3}(3.1,1.2)
 \roundnode{n4}(2.1,.4)
 \diredge{n1}{n2}
 \diredge{n2}{n3}
 \diredge{n3}{n4}
 \diredge{n4}{n1}
 \diredge{n1}{n3}
 \diredge{n2}{n4}
 \diredge{n4}{n2}
 \diredge{n3}{n1}
 \diredge{n2}{n1}
 \diredge{n3}{n2}
 \diredge{n4}{n3}
 \diredge{n1}{n4}
\end{graph}
\begin{graph}(.7,2.4)
 \freetext(.35,1.2){\large $\DSldder^4$}
\end{graph}
\begin{graph}(3.4,2.4)(.4,0)
 \roundnode{n1}(1.1,1.2)
 \roundnode{n2}(2.1,2)
 \roundnode{n3}(3.1,1.2)
 \roundnode{n4}(2.1,.4)
 \diredge{n1}{n2}
 \diredge{n2}{n3}
 \diredge{n3}{n4}
 \diredge{n4}{n1}
 \diredge{n1}{n3}
 \diredge{n2}{n4}
 \diredge{n4}{n2}
 \diredge{n3}{n1}
\end{graph}
\begin{graph}(.7,2.4)
 \freetext(.35,1.2){\large $\ldder$}
\end{graph}
\begin{graph}(3.4,2.4)(.4,0)
 \roundnode{n1}(1.1,1.2)
 \roundnode{n2}(2.1,2)
 \roundnode{n3}(3.1,1.2)
 \roundnode{n4}(2.1,.4)
 \diredge{n1}{n2}
 \diredge{n2}{n3}
 \diredge{n3}{n4}
 \diredge{n4}{n1}
 \diredge{n1}{n3}
 \diredge{n2}{n4}
 \diredge{n3}{n1}
\end{graph}
\end{tabular}

\end{center}
\caption{An execution of \ttt{TransClosure} \label{fig:TransClosure_execution}}
\end{figure}

\begin{proposition}[Termination]
On every host graph $G$, program \ttt{TransClosure} terminates after at most\/ $|V_G| \times |V_G|$ rule applications. 
\end{proposition}
\begin{proof} 
Given any host graph $X$, let 
\[ \#X = |\{\tuple{v,w} \in V_X \times V_X \mid \text{there is no edge from $v$ to $w$} \}|.\] 
Note that $\#X \leq |V_X| \times |V_X|$. By the application condition of \ttt{link}, every step $G \dder_{\mtt{link}} H$\/ satisfies $\#H = \#G - 1$. Hence $\mtt{link!}$ terminates after at most $|V_G| \times |V_G|$ rule applications. 
\end{proof}

\begin{proposition}[Correctness]
Program \ttt{TransClosure} returns the transitive closure of the input graph. 
\end{proposition} 
\begin{proof} 
The previous proposition guarantees that for every input graph $G$, the loop \ttt{link!} returns some host graph $T$. Because \ttt{link} does not delete or relabel any items, each step $X \dder_{\mtt{link}} Y$\/ comes with an injective graph morphism $X \to Y$. It follows that $T$\/ is an extension of $G$ (up to isomorphism).

We show that $T$\/ is transitive by induction on the length of paths in $T$. Consider a directed path $v_0, v_1, \dots, v_n$ with $v_0 \neq v_n$. Without loss of generality, we can assume that $v_0,\dots,v_n$ are distinct. If $n=1$, there is nothing to show because there is an edge from $v_0$ to $v_1$. Assume now $n>1$. By induction hypothesis, there is an edge from $v_0$ to $v_{n-1}$. Thus there exist edges $v_0 \to v_{n-1} \to v_n$ in $T$. Since \ttt{link} has been applied as long as possible, there must exist an edge from $v_0$ to $v_n$. (If there was no such edge, \ttt{link} would be applicable to $T$.)

Finally, $T$\/ is the smallest transitive extension of $G$ by the following invariant: given any derivation $G \der_{\mtt{link}} H$\/ and any edge $v \to v'$ in $H$\/ created by the derivation, there is no such edge in $G$ but a path from $v$ to $v'$. This invariant is shown by a simple induction on the length of derivations $G \der_{\mtt{link}} H$. 
\end{proof}

\section{Case Study: Vertex Colouring}
\label{sec:colouring}

A \emph{vertex colouring}\/ for a graph $G$ is an assignment of colours to $G$'s nodes such that adjacent nodes have different colours. As common in the literature \cite{Cormen-Leiserson-Rivest-Stein09a}, we use positive integers as colours. The program \ttt{Colouring} in Figure \ref{fig:colouring} colours an input graph by adding an integer to each node's label. Initially, each node gets colour 1 by the loop \ttt{init!}. To prevent repeated applications of \ttt{init} to the same node, nodes are first shaded and then unmarked by \ttt{init}. Once all nodes have got colour 1, the loop \ttt{inc!} repeatedly increments the target colour of edges that have the same colour at source and target. This continues as long as there are pairs of adjacent nodes with the same colour. Note that this process is highly non-deterministic and may result in different colourings; for example, Figure \ref{fig:colour_results} shows two different outcomes for the same input graph. (Finding a colouring with a minimal number of colours is an NP-complete problem \cite{Cormen-Leiserson-Rivest-Stein09a} and requires a more complicated program.)

\begin{figure}[htb]
 \begin{center}
  \input{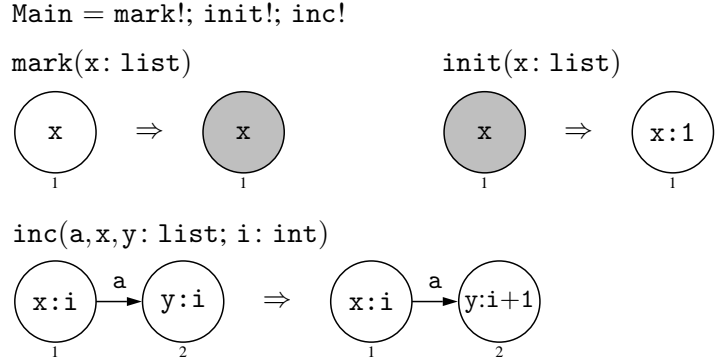}
 \end{center}
\caption{The program \ttt{Colouring}}
\label{fig:colouring}
\end{figure}

\begin{figure}[htb]
\begin{center}
 \input{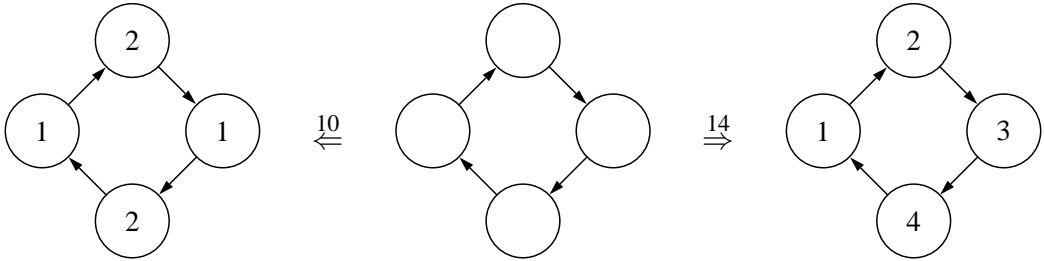} 
\end{center}
\caption{Different results from executing \ttt{Colouring}}
\label{fig:colour_results}
\end{figure}

The partial correctness of \ttt{Colouring} is easy to establish, we get it essentially for free from the meaning of the loop construct `\ttt{!}'. We will address termination afterwards.

\begin{proposition}[Partial correctness]
Given any input graph $G$, if \ttt{Colouring} terminates then it returns $G$ correctly coloured.  
\end{proposition}
\begin{proof}
Let $ G \dder^*_{\mtt{mark}} G' \dder^*_{\mtt{init}} H \dder^*_{\mtt{inc}} M$\/ be an execution of \ttt{Colouring} on $G$. Then $H$\/ is obtained from $G$ by replacing each node label $x$ with $x{:}1$. If $M$\/ is not correctly coloured, then it must contain adjacent nodes with the same colour. Hence \ttt{inc} would be applicable to $M$, but the meaning of `\ttt{!}' implies that $M$\/ results from applying \ttt{inc} as long as possible. 
\end{proof}

Proving that \ttt{Colouring} terminates is more challenging, we first establish an invariant of \ttt{inc}. Given a node $v$\/ with a label of the form $x\,{:}i$, $i \in \N$, we denote $i$\/ by $\mrm{colour}(v)$; for a host graph $G$ with coloured nodes, we define $\mrm{Colours}(G) = \{\mrm{colour}(v) \mid v \in V_G\}$.

\begin{lemma}[Invariant]
\label{lem:inc_invariant}
Consider any derivation $G \der_{\mathtt{inc}} H$ with\/ $\mrm{Colours}(G) = \{1\}$. Then $\mrm{Colours}(H) = \{i \mid 1 \leq i \leq n\}$ for some $1 \leq n \leq |V_H|$.
\end{lemma}
\begin{proof}
For every step $X \dder_{\mathtt{inc}} Y$, we have $\mrm{Colours}(Y) = \mrm{Colours}(X) \cup \Delta$ where $\Delta$ is either empty or $\{\max(\mrm{Colours}(X)) + 1\}$. The invariant follows then by induction on the length of $G \der_{\mathtt{inc}} H$.
\end{proof}

Given any coloured host graph $X$, define $\#X = \sum_{v \in V_X} \mrm{colour}(v)$. Then the invariant of Lemma \ref{lem:inc_invariant} provides an upper bound for $\#H$, viz.\ $\#H \leq 1 + 2 + \dots + |V_H|$. We exploit this in the following proof.

\begin{proposition}[Termination]
On every input graph $G$, \ttt{Colouring} terminates after $\mathrm{O}(|V_G|^2)$ rule applications. 
\end{proposition}
\begin{proof}
It is clear that the loops \ttt{mark!} and \ttt{init!} terminate: the first decreases in each step the number of unmarked nodes, the second decreases in each step the number of marked nodes.

To show that \ttt{inc!} is terminating, consider a coloured graph $G$ with $\mrm{Colours}(G) = \{1\}$ and suppose that there is an infinite derivation $G=G_0 \dder_{\mtt{inc}} G_1 \dder_{\mtt{inc}} G_2 \dder_{\mtt{inc}} \dots$ Then, by the labelling of \ttt{inc}, we have $\#G_i < \#G_{i+1}$ for every $i \geq 0$. However, Lemma \ref{lem:inc_invariant} implies that for all $i \geq 0$,
\[ \#G_i \leq \sum_{j=1}^{|V_{G_i}|} j = \sum_{j=1}^{|V_{G}|} j \]
where $|V_{G_i}| = |V_{G}|$ holds because \ttt{inc} preserves the number of nodes. Thus the infinite derivation cannot exist. Moreover, because the upper bound for the values $\#G_i$ is quadratic in $|V_G|$, any sequence of \ttt{inc} applications starting from $G$ cannot contain more than $\mrm{O}(|V_G|^2)$ rule applications. 
\end{proof}

\section{Case Study: Cycle Checking}
\label{sec:cyclic}

Our third example program shows how to test the input graph for a property and then continue computing with the same graph. The program in Figure \ref{fig:cyclic} checks whether a host graph $G$ contains a directed cycle and then, depending on the result, executes either program $P$\/ or program $Q$ on $G$. 

\begin{figure}[htb]
 \begin{center}
     \graphnodesize{1}
   \grapharrowlength{0.2}
   \grapharrowwidth{0.6}
   \autodistance{1.0}
   \graphnodecolour{1}

\begin{tabular}{l}
\ttt{Main = if Cyclic then $P$\/ else $Q$}\\
\ttt{Cyclic = delete!; \{edge,\! loop\}}
\\\\

$\mtt{delete(a,x,y\colon list)}$
\\[.5ex]
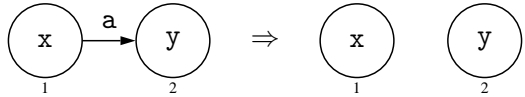
\begin{tabular}{rcl}
\begin{graph}(2.5,1)(0,-.1)
\roundnode{A}(0.4,0.4)
\autonodetext{A}{$\x$}
\roundnode{B}(2.1,0.4)
\autonodetext{B}{$\y$}
\autonodetext{A}[s]{\tiny 1}
\autonodetext{B}[s]{\tiny 2} 
\diredge{A}{B}
\bowtext{A}{B}{.15}{$\mtt{a}$}
\end{graph}
&
\begin{graph}(.8,1)(0,-.1)
\freetext(0.4,0.4){$\dder$}
\end{graph}
&
\begin{graph}(2.5,1)(0,-.1)
\roundnode{A}(0.4,0.4)
\autonodetext{A}{$\x$}
\roundnode{B}(2.1,0.4)
\autonodetext{B}{$\y$}
\autonodetext{A}[s]{\tiny 1}
\autonodetext{B}[s]{\tiny 2} 
\end{graph}
\\[1ex]
\multicolumn{3}{l}{\ttt{where $\mtt{indeg(1)=0}$}}
\end{tabular}
\\\\

$\mtt{edge(a,x,y\colon list)}$
\\[.5ex]
\begin{tabular}{rcl}
\begin{graph}(2.5,1)(0,-.1)
\roundnode{A}(0.4,0.4)
\autonodetext{A}{$\x$}
\roundnode{B}(2.1,0.4)
\autonodetext{B}{$\y$}
\autonodetext{A}[s]{\tiny 1}
\autonodetext{B}[s]{\tiny 2} 
\diredge{A}{B}
\bowtext{A}{B}{.15}{$\mathtt{a}$}
\end{graph}
&
\begin{graph}(.8,1)(0,-.1)
\freetext(0.4,0.4){$\dder$}
\end{graph}
&
\begin{graph}(2.5,1)(0,-.1)
\roundnode{A}(0.4,0.4)
\autonodetext{A}{$\x$}
\roundnode{B}(2.1,0.4)
\autonodetext{B}{$\y$}
\autonodetext{A}[s]{\tiny 1}
\autonodetext{B}[s]{\tiny 2} 
\diredge{A}{B}
\bowtext{A}{B}{.15}{$\mathtt{a}$}
\end{graph}
\end{tabular}
\\\\

$\mtt{loop(a,x\colon list)}$
\\[.5ex]
\begin{tabular}{rcl}
\begin{graph}(1.65,1)(0,-.1)
\roundnode{A}(.4,.4)
\autonodetext{A}{$\x$}
\autonodetext{A}[s]{\tiny 1}
\dirloopedge{A}(.54,.2)(.54,-.2)
\freetext(1.55,.4){$\mathtt{a}$}
\end{graph}
&
\begin{graph}(.8,1)(0,-.1)
\freetext(0.4,0.4){$\dder$}
\end{graph}
&
\begin{graph}(1.65,1)(0,-.1)
\roundnode{A}(.4,.4)
\autonodetext{A}{$\x$}
\autonodetext{A}[s]{\tiny 1}
\dirloopedge{A}(.54,.2)(.54,-.2)
\freetext(1.55,.4){$\mathtt{a}$}
\end{graph}
\end{tabular}
\end{tabular}
 \end{center}
\caption{The program \ttt{CycleCheck}}
\label{fig:cyclic}
\end{figure}

The presence of cycles is checked by deleting, as long as possible, edges whose source nodes have no incoming edges. To ensure the latter, rule \texttt{delete} uses GP 2's built-in function \ttt{indegree} which returns the number of edges going into a node.  When \texttt{delete} is no longer applicable, the resulting graph contains edges if and only if the input graph is cyclic. For example, Figure \ref{fig:CycleCheck_execution} shows two executions of \ttt{CycleCheck}, the left on a cyclic input graph and the right on an acyclic graph. (We also show an intermediate graph for each execution.)

\begin{figure}[htb]
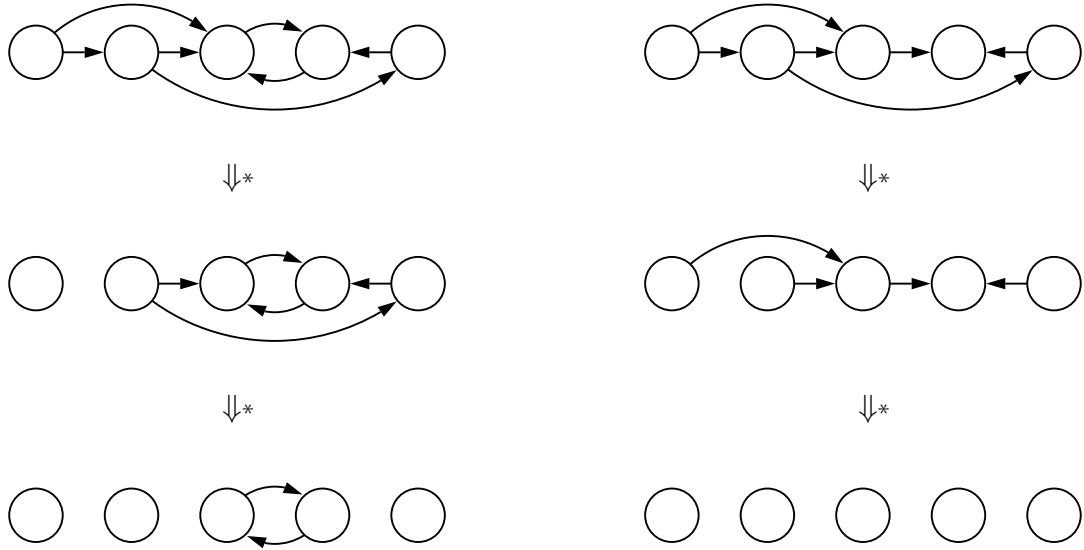

\begin{center}
\begin{tabular}{ccc}
\input{Graphs/cyclic.graph} & \hspace{1cm} & \input{Graphs/acyclic.graph}
\\
\input{Graphs/downarrow.graph} && \input{Graphs/downarrow.graph}
\\
\input{Graphs/cyclic1.graph} && \input{Graphs/acyclic1.graph}
\\
\input{Graphs/downarrow.graph} && \input{Graphs/downarrow.graph}
\\
\input{Graphs/cyclic2.graph} && \input{Graphs/acyclic2.graph}
\end{tabular}
\end{center}
\caption{Two executions of \ttt{CycleCheck} \label{fig:CycleCheck_execution}}
\end{figure}

The correctness of this method for cycle checking relies on the following invariant.

\begin{lemma}[Invariant]
\label{lem:delete_invariant}
Given any step $G \dder_{\mtt{delete}} H$, graph $G$ is cyclic if and only if graph $H$ is cyclic.
\end{lemma}
\begin{proof}
Suppose that $G$ contains a directed cycle. By the application condition of \texttt{delete}, none of the edges on the cycle are deleted. Conversely, if $G$ is acyclic, deleting an edge cannot create a cycle.
\end{proof}

We also need the following property of acyclic graphs, which is easy to prove \cite{Bang_Jensen-Gutin09a}.

\begin{lemma}[Acyclic graphs]
\label{lem:acyclic-graphs}
Every non-empty acyclic graph contains a node without incoming edges.
\end{lemma}

For stating the correctness of \ttt{CycleCheck}, we use GP 2's semantic function $\Sem{\_}$ which maps each input graph to the set of possible execution outcomes (see \cite{Plump12a}).

\begin{proposition}[Correctness]
For every host graph $G$, 
\[ \Sem{\mtt{CycleCheck}}G = \left\{ \begin{array}{cl}
                                      \Sem{P}G & \text{if $G$ is cyclic,}\\
                                      \Sem{Q}G & \text{otherwise.}
                                     \end{array} \right.
\]
\end{proposition}
\begin{proof}
We show that running procedure \ttt{Cyclic} on $G$ returns some graph if $G$ is cyclic, and fails otherwise. By the meaning of the if-then-else statement, this implies the proposition.

Since each application of \ttt{delete} reduces graph size, executing the loop \ttt{delete!} on $G$ results in some graph $H$.

\emph{Case 1:} $G$ is cyclic. Then, by Lemma \ref{lem:delete_invariant} and a simple induction on the derivation $G \der_{\mtt{delete}} H$, graph $H$\/ is also cyclic. Hence $H$\/ contains at least one edge, implying that $\{\mtt{edge},\mtt{loop}\}$ is applicable to $H$.

\emph{Case 2:} $G$ is acyclic. With Lemma \ref{lem:delete_invariant} follows that $H$\/ is acyclic, too. We show that $H$\/ does not contain edges. Suppose that $H$\/ contains some edges (which cannot be loops). Consider the subgraph $S$\/ of $H$\/ obtained by removing all isolated nodes. Then, by Lemma \ref{lem:acyclic-graphs}, there is a node $v$ in $S$\/ without incoming edges. Since $v$ is not isolated, it must have an outgoing edge $e$. But then \ttt{delete} is applicable to $e$, contradicting the fact that \ttt{delete} is not applicable to $H$.

Thus $H$\/ is edge-less and hence $\{\mtt{edge},\mtt{loop}\}$ fails on $H$.
\end{proof}

\section{Case Study: Series-Parallel Graphs}
\label{sec:series-parallel}

Our final case study is the recognition of series-parallel graphs. This graph class was introduced in \cite{Duffin65a} as a model of electrical networks and comes with an inductive definition:
\begin{itemize}
\item Every graph $G$ consisting of two distinct nodes $v_1,\,v_2$ and an edge from $v_1$ to $v_2$ is series-parallel. Define $\mrm{source}(G) = v_1$ and $\mrm{sink}(G) = v_2$.
\item Given series-parallel graphs $G$ and $H$, each of the following operations yields a series-parallel graph when applied to the disjoint union $G+H$:
\begin{itemize}
\item Series composition: Merge $\mrm{sink}(G)$ with $\mrm{source}(H)$. Define $\mrm{source}(G)$ to be the new source and $\mrm{sink}(H)$ to be the new sink.
\item Parallel composition: Merge $\mrm{source}(G)$ with $\mrm{source}(H)$ and $\mrm{sink}(G)$ with $\mrm{sink}(H)$. Define the merged source nodes to be the new source and the merged sink nodes to be the new sink.
\end{itemize}
\end{itemize}
Series-parallel graphs can be characterised by two reduction operations on graphs. On GP 2 graphs, these operations are equivalent to applying one of the rules \ttt{series} and \ttt{parallel} from Figure \ref{fig:series-parallel}. Hence we can state the characterisation in terms of derivations with the rule set $\mtt{Reduce} = \{\mtt{series},\,\mtt{parallel}\}$.

\begin{proposition}[\cite{Duffin65a}]
\label{prop:characterisation}
A host graph $G$ is series-parallel if and only if there is a derivation $G \der_{\mrm{Reduce}} E$, where $E$ is a series-parallel graph consisting of two nodes and one edge.
\end{proposition}

Given a host graph $G$, the procedure \ttt{Series-parallel} in Figure \ref{fig:series-parallel} applies the reduction rules as long as possible to obtain some reduced graph $H$. (Termination is guaranteed because both rules decrease graph size.) To check whether $H$ has the shape of graph $E$ from Proposition \ref{prop:characterisation}, the procedure attempts to delete a subgraph of $E$'s shape and then tests if there is any remaining node. If the deletion fails or a remaining node is detected, $H$ does not have the shape of $E$ and the procedure fails. Otherwise, if rule \ttt{delete} succeeds and rule \ttt{nonempty} fails, the procedure succeeds by returning the empty graph.

\begin{figure}[htb]
 \begin{center}
  \input{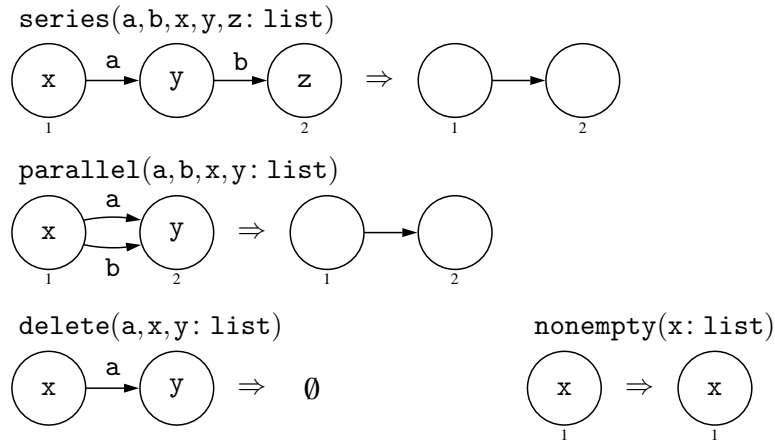}
 \end{center}
\caption{The program \ttt{SeriesParallel}}
\label{fig:series-parallel}
\end{figure}

Proposition \ref{prop:characterisation} alone does not guarantee the correctness of SeriesParallel. This is because if the non-deterministic application of \ttt{Reduce} does not end in $E$, there might be some other reduction sequence ending in $E$. We exclude this possibility by showing that any complete reduction of a host graph ends in a graph that is unique up to isomorphism.

\begin{proposition}[Uniqueness of reduced graphs]
\label{prop:uniqueness}
Consider reductions $H_1 \lder_{\mrm{Reduce}} G \der_{\mrm{Reduce}} H_2$ of some host graph $G$. If $H_1$ and $H_2$ are irreducible, then they are isomorphic.
\end{proposition}
\begin{proof} 
We show that \ttt{Reduce} is \emph{confluent}\/ which implies the claim. To employ the technique of \cite{Plump05a}, we analyse the \emph{critical pairs}\/ of all rules obtained from \ttt{series} and \ttt{parallel} by replacing variables with constant lists. There are no critical overlaps between \ttt{series} and \ttt{parallel}, hence all critical pairs are self-overlaps of either \ttt{series} or \ttt{parallel}. There are two types of critical pairs of \ttt{parallel}, which are easily shown to be \emph{strongly joinable} in the terminology of \cite{Plump05a}. There are also two types of critical pairs of \ttt{series}. We consider one of them:
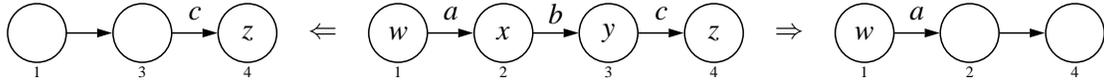
\begin{center}
 \input{Graphs/critical-pair.graph}
\end{center}
where $a$ to $c$ and $w$ to $z$ are arbitrary host graph lists. This pair is strongly joinable by the following applications of \ttt{series}:
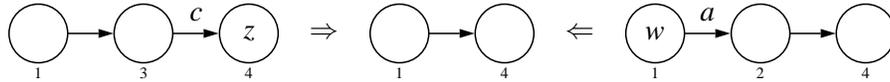
\begin{center}
 \input{Graphs/joining-pair.graph}
\end{center}
The other critical pair of \ttt{series} is obtained from the above pair by merging nodes 1 and 4 in all three graphs. Then the outer graphs are isomorphic in the way required by strong joinability. With \cite{Plump05a} follows that \ttt{Reduce} is confluent.
\end{proof}

\begin{proposition}[Correctness]
For every host graph $G$, running \ttt{Series-parallel} on $G$ returns the empty graph if $G$ is series-parallel and fails otherwise.
\end{proposition}
\begin{proof}
Let $H$ be the graph resulting from running \ttt{Reduce} on $G$. By Proposition \ref{prop:uniqueness}, $H$ is determined uniquely up to isomorphism. Thus, if $G$ is series-parallel, Proposition \ref{prop:characterisation} implies that $H$ consists of two nodes and an edge between them. It follows that \ttt{Series-parallel} returns the empty graph.

If $G$ is not series-parallel, Proposition \ref{prop:characterisation} implies that $H$ has some other shape. Then either rule \ttt{delete} is not applicable or \ttt{delete} is applicable but afterwards \ttt{nonempty} is applicable. In both cases \ttt{Series-parallel} fails, as required. 
\end{proof}

\section{Conclusion}
\label{sec:conclusion}

In this paper, we show by some case studies that GP 2's graph transformation rules allow high-level reasoning on partial and total correctness of graph programs. Instead of employing a rigorous formalism such as the Hoare-logic in \cite{Poskitt-Plump12a}, we use standard mathematical language for expressing assertions and proofs. This both increases the expressive power of specifications and frees programmers from the notational constraints of a formal verification calculus. The main mathematical tool in our sample proofs is induction over derivation sequences of graph transformation rules.

We do not propose to abandon rigorous program verification in favour of a more liberal approach to justifying correctness. In contrast, proof calculi with precise syntax and semantics are indispensable for achieving confidence in verified software. We view this paper's approach as complementary in that it allows programmers to reason about graph programs without getting stuck in formal detail. Future work should address the question how to refine proofs in ordinary mathematical language into rigorous proofs in verification calculi.

\bibliographystyle{eptcs}
\bibliography{/home/djp10/Bibtex/abbr,%
              /home/djp10/Bibtex/absredu.bib,%
              /home/djp10/Bibtex/trs,%
              /home/djp10/Bibtex/gragra,%
              /home/djp10/Bibtex/graph-algorithms,%
              /home/djp10/Bibtex/graphs-misc,%
              /home/djp10/Bibtex/gratra-languages,%
              /home/djp10/Bibtex/graph-matching,%
              /home/djp10/Bibtex/proglang,%
              /home/djp10/Bibtex/pointers,%
              /home/djp10/Bibtex/verification,%
              /home/djp10/Bibtex/misc}

\end{document}